\documentclass{birkjour}
 \newtheorem{thm}{Theorem}[section]
 
 \newtheorem{lem}[thm]{Lemma}
 
 \theoremstyle{definition}
 
 \theoremstyle{remark}
 \newtheorem{rem}[thm]{Remark}
 
 \numberwithin{equation}{section}
\usepackage{url}

\usepackage{subfig}
\usepackage{paralist,graphics,epsfig,graphicx,epstopdf,mathrsfs}
\usepackage{float,color,ulem,comment,tabulary,booktabs}
\newfloat{figure}{H}{lof}
\newfloat{table}{H}{lot}
\floatname{figure}{\figurename}
\floatname{table}{\tablename}

\usepackage{subfig} 
\usepackage{caption}
\usepackage{mathtools}

\renewcommand{\a }{\alpha }
\renewcommand{\b }{\beta }

\newcommand{\g }{\gamma}

\renewcommand{\L }{\Lambda }

\renewcommand{\O }{\Omega }

\newcommand{\be}{\begin{equation}}
\newcommand{\ee}{\end{equation}}

 \usepackage{mathrsfs}

\renewcommand\theenumi{\roman{enumi}}

\renewcommand{\theequation}{\thesection.\arabic{equation}}

\renewcommand{\epsilon}{\varepsilon}

\begin{document}

\title[Analysis of COVID-19 evolution in Senegal: impact of health care capacity]{Analysis of COVID-19 evolution in Senegal: impact of health care capacity}

\author[Fall M.M.]{Mouhamed M. Fall} 
\address{The African Institute for Mathematical Sciences (AIMS)\br
Mbour, Senegal.}
\email{mouhamed.m.fall@aims-senegal.org}
\author[Ndiaye B.M.]{Babacar M. Ndiaye}
\address{Laboratory of Mathematics of Decision and Numerical Analysis\br
University of Cheikh Anta Diop.\br
BP 45087, 10700. Dakar, Senegal.}
\email{babacarm.ndiaye@ucad.edu.sn}
\author[Seydi O.]{Ousmane Seydi}
\address{D\'epartement Tronc Commun, \br
 \'Ecole Polytechnique de Thi\`es, Senegal.}
\email{oseydi@ept.sn}
\author[Seck D.]{Diaraf Seck}
\address{Laboratory of Mathematics of Decision and Numerical Analysis\br
University of Cheikh Anta Diop, Dakar, Senegal.\br
IRD, UMMISCO, Dakar, Senegal.}
\email{diaraf.seck@ucad.edu.sn}

\thanks{This work was completed with the support of the NLAGA project}

\date{November 09, 2020}
\begin{abstract}
We consider a compartmental model from which we incorporate  a time-dependent health care capacity having a logistic growth. This allows us to take into account the Senegalese authorities response in anticipating the growing number of infected cases. We highlight the importance of anticipation and timing to avoid overwhelming that could impact considerably the treatment of patients  and the well-being of health care workers. A condition, depending on the health care capacity and the flux of new hospitalized individuals, to avoid possible overwhelming is provided. We also use machine learning approach to project forward the cumulative number of cases from March 02, 2020, until  1st December, 2020.
\end{abstract}
\maketitle

\section{Introduction}\label{intro}
COVID-19, declared a pandemic by the World Health Organization (WHO) \cite{WHO1} on 11 March 2020, is still spreading around the world up to date 26 September 2020. The number of people infected is beyond 32 million on 26 September 2020 with 989,380 deaths \cite{Wiki}. In Senegal, the number of cumulative cases is currently 14839 with 2624 individuals undergoing treatment on 25 September 2020\cite{msas}. The first cases, from Wuhan, were notified to WHO on 31 December 2019 \cite{WHO1,WHO2} while, Senegal notified its first case on 02 March 2020  \cite{msas}. Because of its limited resources, as in many sub-Saharan African countries, it is therefore valuable to understand the growth and the timing in responding to the   logistic needs of their health system.  We note that  several developed countries that nevertheless have high-capacity health structures have been overwhelmed and this considerably impacted negatively  in the  combat against the pandemic. In \cite{Moghadas} it is discovered that the growing number of COVID-19 cases in the United States could gravely challenge the critical care capacity, thereby exacerbating case fatality rates. As pointed out in \cite{Weissman}, under-resourced health systems may pose a threat to patient care as well as the safety and well-being of health care workers. \\
From the appearance of the first cases of COVID-19, the Senegalese authorities  decided   to  anticipate by increasing gradually   the capacity of its hospitals to take care of patients suffering from COVID-19.  One of  the aims of the  present paper is  to study different scenarios of involution of the number of virus patients  by varying the  growth rate  and the number of, say, beds or rooms.  
 As an approach, we use a mathematical model that integrates the different stages of individuals (Susceptible, Exposed, Asymptomatic, Symptomatic, Removed) with the parameters depending on the health resources availability. 
To this cope, we consider a modified SEWIR model \cite{ndiaye2} in which we include a time-dependent carrying capacity $K(t)$ representing the evolution of the Senegalese hospitals' capacity in response to the growing number of cases in Senegal. 
Our analysis highlights the importance of anticipation and timing to avoid overwhelming in the health system. More precisely, our analysis shows, in particular, that, while the epidemic growths exponentially fast,  by respecting a certain logistic  growth (at which the number of beds is being added) in the health system,  the epidemic might remain under control. As a consequence, a non-negligible number of  expenses can be saved.\\
Finally, a machine learning approach to project forward the cumulative number of  virus infected cases is provided.  Let us mention that several studies have been conducted in different contexts to assess the impact of hospital saturation in COVID-19. We refer for examples to  \cite{Cavallo,Demasse,IHME,Moghadas,Richard} and the references therein.
Our work complements several studies that have been carried out on the spread of Sars-CoV-2 in Senegal \cite{ndiaye4,ndiaye3,balde,Diaby,ndiaye1,ndiaye2,sarr,samb}. It should be noted that the aforementioned  studies did not address the impact of saturation on patients care in Senegal. However, some studies on the prediction of the number of new cases for Senegal have already been carried out using machine learning \cite{ndiaye4,ndiaye3,ndiaye1,ndiaye2} and on the impact of contaminated objects using compartmental models \cite{Diaby}.\\
\noindent The paper is organized as follows. In section \ref{model}, we describe in detail our model as well as the parameters involved. It is followed by Section \ref{Sec3} concerning the estimation of the unknown parameters. In section \ref{Sec-Eff} we derive the time-dependent effective reproduction number which is a key function that informs about the impact of the various decisions taken by the authorities during the epidemic. In Section \ref{Sec-Num},  we first present the calibration of the model to the cumulative number of reported cases. We then present numerical simulations on the impact of the dynamics of health capacity on the patients care. Finally, in Section \ref{mlprophet}, we use machine learning to project forward the number of cumulative cases until 1st December 2020.

\section{Model description}\label{model}
\noindent We consider compartmental epidemic model that incorporates a time-dependent carrying capacity $K(t)$ of health structures for the treatment of COVID-19 patients at the hospitals. The model reads as follow 
\begin{equation}\label{eq:SEWIR-F}
\left\{\begin{array}{llll}
S'(t)&=&-\b_1(t)[W(t)+\epsilon I(t)]S(t)\vspace{0.2cm}\\
E'(t)&=&\b_1(t)[W(t)+\epsilon I(t)]S(t)-\b_2 E(t)\vspace{0.2cm}\\
W'(t)&=&\b_2 E(t)-\b_3 W(t)\vspace{0.2cm}\\
I'(t) &=&(1-\a_1)\b_3 W(t)-\gamma I(t) \left(1-I(t)/K(t) \right)_+ - \eta I(t)\vspace{0.2cm}\\
R'(t) &=&\a_1 \b_3 W(t)+\gamma I(t) \left(1-I(t)/K(t) \right)_+ + \eta I(t)
\end{array}
\right.
\end{equation}
with initial conditions 
\begin{equation}
S(t_0)=S_0,\ E(t_0)=E_0,\ W(t_0)=W_0,\ I(t_0)=I_0,\ R(t_0)=R_0
\end{equation}
and
\begin{equation}
\left(1-I(t)/K(t) \right)_+=\left\lbrace
\begin{array}{lllll}
1-I(t)/K(t) & \text{ if } \ & I(t)\leq K(t) \vspace{0.2cm}\\
0 &  \text{ if } \ & I(t) > K(t).
\end{array}
\right.
\end{equation}
The state variable $S(t)$ is the number of susceptible individuals at time $t$,  $E(t)$ the number of exposed individuals at time $t$ and $W(t)$ the number of waiting cases for confirmation (with infectiousness) at time $t$. The number of confirmed and monitored cases at time $t$ is denoted by $I(t)$. The variable  $K(t)$ describes the effective carrying capacity of the health structure, which is the maximum number of $I(t)$ individuals that can be correctly monitored. Finally $R(t)$ represents the removed individuals at time $t$ due to death or recovery. The figure below represents the diagram flux of the model.
\begin{figure}[h!]
\begin{center}
\includegraphics[scale=0.65]{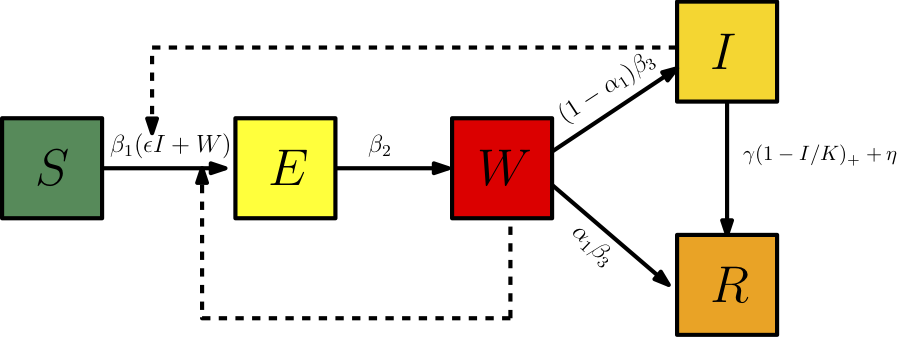}
\end{center}
\par\vspace{-0.25cm}
\caption{Flow diagram of the model.}\label{fig1:diagram}
\end{figure}
\noindent The term $\epsilon \beta_1(t) I(t) S(t)$ represents the transmission of the disease within the health structures where $\epsilon>0$ can be interpreted   as the probability that a susceptible individual is working in the   health system. More precisely we set $\epsilon=N_H/S_0$ with $N_H$ the number of health care workers and $S_0$ the initial number of susceptible individuals.  The quantity of $1/\beta_2$ is the average incubation period. We assume that the average waiting time for confirmation is $1/\beta_3$ days.  After the waiting time, a fraction $1-\alpha_1$ is monitored while the remaining fraction $\alpha_1$ is removed due to death or  recovered. The removal rate of the confirmed cases is described by $t\rightarrow \gamma \left(1-I(t)/K(t) \right)_+ + \eta$ in order to take into account the saturation effect of the public health structures. The term $\gamma \left(1-I(t)/K(t) \right)_++\eta$ which is a decreasing function of $I(t)$ describes the efficiency of the health structures. The maximal efficiency is $\gamma+\eta$ while the minimal efficiency $\eta$  is reached when the number of confirmed cases becomes greater than the carrying capacity $K(t)$. In our modeling procedure, we have introduced a time-dependent carrying capacity which is motivated by the fact that in Senegal, the carrying capacity of  the health structures was gradually increased in response to the increase in the number of infected individuals. The dynamics of the increase in the number of beds remains unknown because only announcements have been made. We chose a logistic equation to mimic the growth of the public health capacity to its saturation. Indeed, from the onset of the epidemic, the authorities announced a carrying capacity of 30 which gradually increased according to the number of cases. The maximum capacity of the structures remains unknown but in view of the announcements made, we assume that it is approximately 3000. The growth rate $r=0.1$ of $t\rightarrow K(t)$ is chosen such that it increases from $30$ on 02 March to approximately $3000$ on 10 June and that the number of hospitalized individuals is below the effective carrying capacity. There are other possibilities of choosing the growth rate $r$ but the value $r=0.1$ appears to give the best fit to the data. The figure below describes the evolution of $K(t)$ with respect to time and the evolution of the number of hospitalized individuals. \\
The evolution of the public health capacity and number of monitored are illustrated in Figure \ref{fig2:Evolcapacity}.\\
\begin{figure}[h!]
\begin{center}
\includegraphics[scale=0.45]{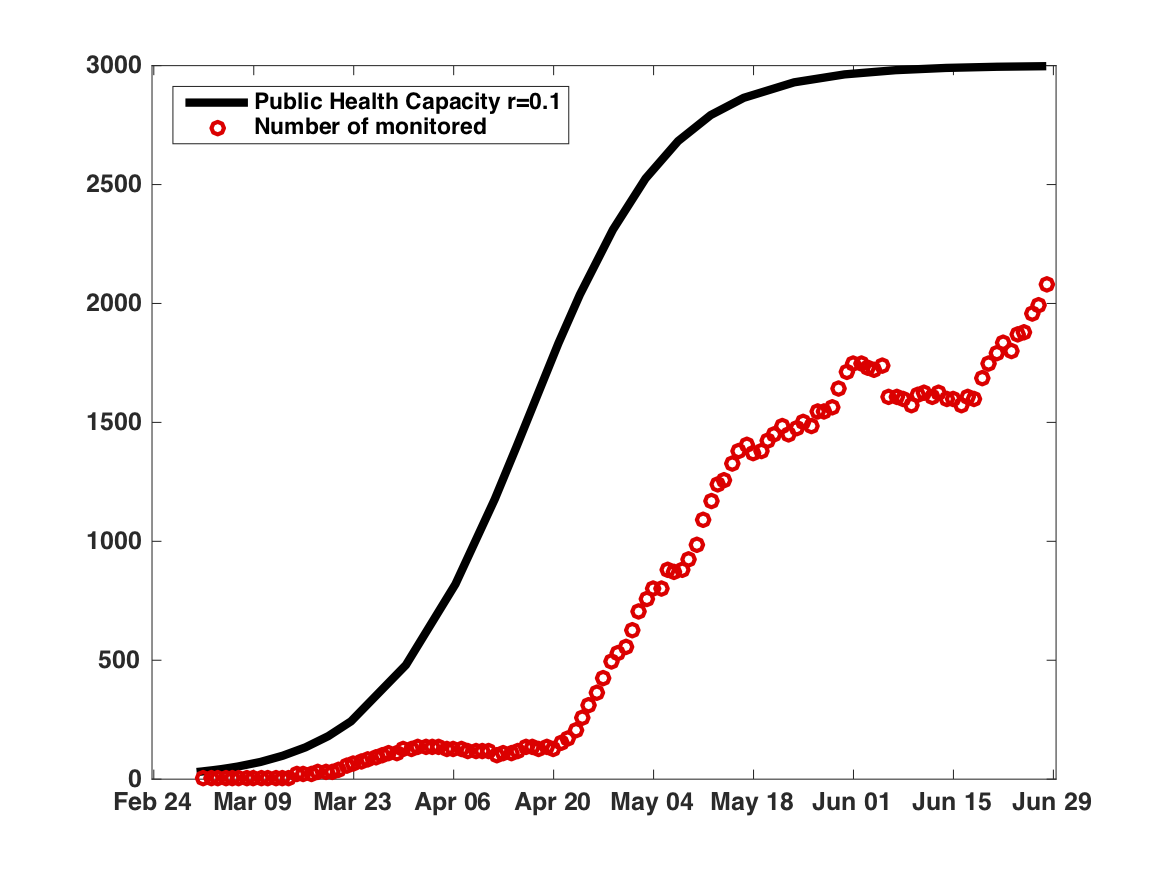}
\end{center}
\par\vspace{-0.25cm}
\caption{Evolution of the hospital carrying capacity and the number of hospitalized individuals.}\label{fig2:Evolcapacity}
\end{figure}
\ \\
\noindent More precisely the dynamic of the carrying capacity is given by: 
\begin{equation}\label{Ksa-sa}
	\dfrac{dK}{dt}=\left\lbrace
	\begin{array}{llll}
		r K(1-K/K_{max}) & \text{ if } & K \leq K_{sa}\\
		0 & \text{ if } & K > K_{sa}\\
	\end{array}
	\right.
\end{equation}
\noindent with 
$$
K_0=30,\ r=0.1 \text{ and } K_{\max}=3000.
$$
The parameter $K_{sa}$ will serve to explore different scenarios in which the efforts of increasing health capacities are aborted. In model \eqref{eq:SEWIR-F}-\eqref{Ksa-sa}, we use a time-dependent transmission rate to take into account the various events that have impacted the spread of disease. Indeed several government actions such as partial lockdown, social distancing, interurban traffic ban, border closure have impacted the progression of the disease. The time-dependent transmission rate has the following form   
\begin{equation}\label{b2.3}
\beta_1(t)=\left\lbrace
\begin{array}{lllll}
\beta_{10} & \text{ if } & t_0\leq t < t_1\\
\max\left(\beta_{10} e^{-\theta_1(t-t_1)}, \ \theta_2 \beta_{10}\right)& \text{ if } & t_1 \leq t < t_2 \\
\min(\beta_{11} e^{\theta_1(t-t_2)}, \ \theta_3 \beta_{10})& \text{ if } & t_2 \leq t < t_3 \\
\max\left(\beta_{12}e^{-\theta_1(t-t_3)},\ \theta_4 \beta_{10} \right) & \text{ if } & t\geq t_3,
\end{array}
\right.
\end{equation}
where the parameters are estimated by using the data on reported cases. More precisely we have 
\begin{equation*}
\beta_{11}=\max\left(\beta_{10} e^{-\theta_1(t_2-t_1)}, \ \theta_2 \beta_{10}\right) \ \text{ and } \ \beta_{12}=\min(\beta_{11} e^{\theta_1(t_3-t_2)}, \  \theta_3 \beta_{10})
\end{equation*}
with $t_0= 02$ March the date at which the first case was reported \cite{msas}, $t_1= 24$ March, $t_2= 12$ April and $t_3= 20$ April. We distinguish three phases. Phase $1$ takes place between the dates $t_0=$ 02 March, the beginning of the epidemic,  until the first government  actions at date $t_1= 24$ March. Phase 2 which begins at $t_1$, the date on which decisions such as public closing,  interurban traffic ban, border closure were taken, is ended at date $t_2= 12$ April. The date $t_2$ corresponds to a sudden increase in the number of cases with a decline on approximately date $t_3= 20$ April where Phase 3 began and cover the period of $20$ April to $28$ June the date at which we stop our study. The cause of the sudden increase in $12$ April remains unknown to us at this time. More accurate data would be needed  in order  to understand the advent of this phenomenon. However, we can follow the evolution of the dynamics of the epidemic using the  time-dependent transmission rate.  The parameter $\theta_1$ describes the intensity of the government actions or sudden events that impacted the transmission rate. We refer to \cite{seydi2} where such interpretation has been used to describes the government actions. The parameters $\theta_2$, $\theta_3$ and $\theta_4$ allow us to measure the impact of each phase on the transmission rate $\beta_{10}$. The values of the parameters $\theta_i$, $i=1,2,3,4$ estimated in Section \ref{Sec3} are 
\begin{equation}\label{b2.4}
\theta_1=5.1366, \    \theta_2=0.1224, \    \theta_3=1.7130, \   \theta_4 =0.356,
\end{equation}
and the estimated value of $\beta_{10}$ in Section \ref{Sec3} is 
\begin{equation}\label{b2.5}
\beta_{10}=3.8593 \times   10^{-8}. 
\end{equation}
The figure below describes the transmission rate $t\rightarrow \beta_1(t)$.
\begin{figure}[h!]
\begin{center}
\includegraphics[scale=0.45]{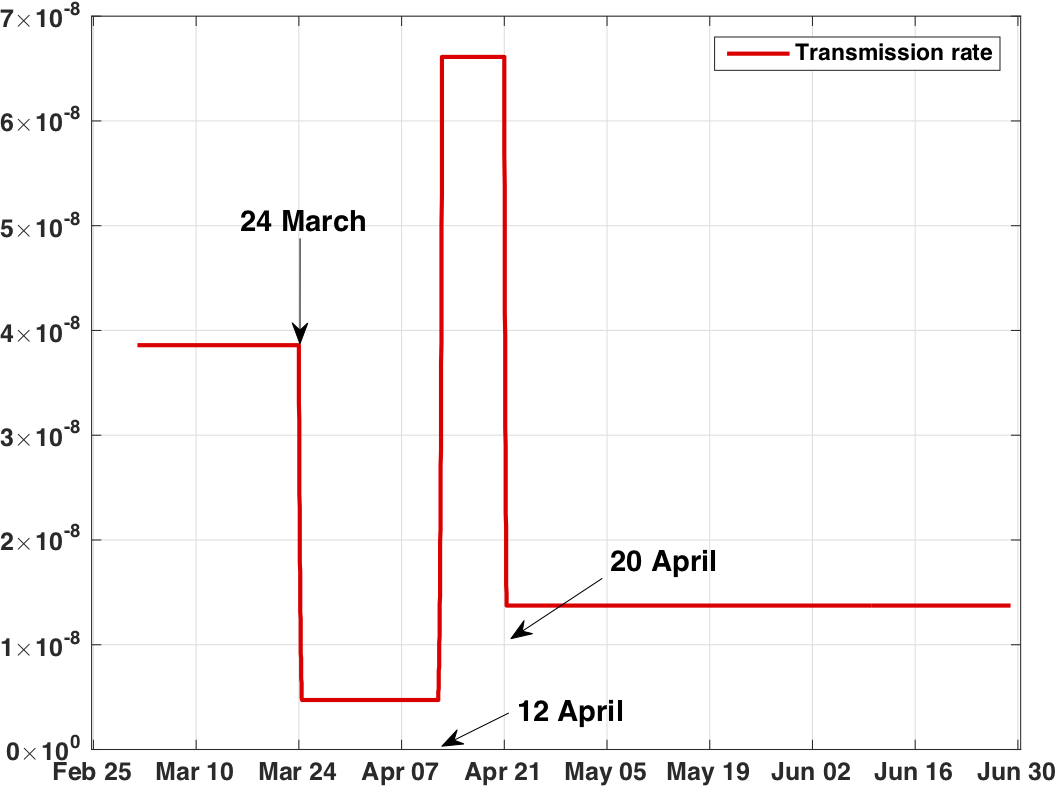}
\end{center}
\par\vspace{-0.25cm}
\caption{Time-dependent transmission $t\rightarrow \beta_1(t)$ of the period from 02 March to 28 June 2020.}\label{fig2:Transmission}
\end{figure}
\noindent Next we summarize the descriptions of the remaining symbols in the table below.
\begin{table}[ht!]
\centering
\caption{Parameters of the model.}\label{tab_parameters}
\begin{tabular}[t]{llllll}
\toprule
{\footnotesize{\textbf{Symbol}}} &{\footnotesize{\textbf{Description}}}& {\footnotesize\textbf{Method}}\\
\midrule
{\footnotesize $t_0 $}
			& {\footnotesize Onset of the epidemic}
			& {\footnotesize fixed}
			\\
			{\footnotesize $S_0$}
		    & {\footnotesize Number of susceptible at time $t_0$ }
		    & {\footnotesize fixed}
			\\
            {\footnotesize$E_0$}
		    & {\footnotesize Number of exposed individuals at time $t_0$ }
		    & {\footnotesize fitted}
			\\
			{\footnotesize $W_0$}
		    & {\footnotesize Number of wating individuals at time $t_0$ }
		    & {\footnotesize fitted}
		    \\
			{\footnotesize $I_0$}
		    & {\footnotesize Number of monitored infectious at time $t_0$ }
		    & {\footnotesize fitted}
		    \\				
			{\footnotesize $ \beta_1(t) $}
			&{\footnotesize Transmission rate}
			&
			{\footnotesize fitted}
			&	
			\\
               {\footnotesize $ 1/\beta_2$}
			& {\footnotesize average latent period}
			&
			{\footnotesize fixed}
			&
			\\

              {\footnotesize $\beta_3$}
			& {\footnotesize Removal rate of the wating individuals }
			& {\footnotesize fixed}
			\\
			{\footnotesize $1-\alpha_1$}
			& {\footnotesize Fraction of wating individuals that become monitored}
			& {\footnotesize fitted }
			\\
			{\footnotesize $ \eta$}
			&  {\footnotesize Rate at which monitored individuals are removed at hospital saturation}
			&  {\footnotesize fixed}
			\\
			{\footnotesize $\gamma+\eta$}
			& {\footnotesize Maximal removal rate of the monitored individuals}
			& {\footnotesize fixed}
			\\
			{\footnotesize $K_{max}$}
			& {\footnotesize Maximal hospital carrying capacity}
			& {\footnotesize fixed}
			\\
			{\footnotesize $r$}
			& {\footnotesize Intrinsic growth rate of the hospital carrying capacity}
			& {\footnotesize fixed} \\
			{\footnotesize $\epsilon$}
			& {\footnotesize Number of healthcare worker per susceptible individual }
			& {\footnotesize fixed} \\
\bottomrule
\end{tabular}
\end{table}
\noindent The values of fixed parameters in Table \ref{tab_parameters} are given in the next table 
\begin{table}
\centering
\caption{{\footnotesize Values of the fixed parameters.}}\label{tab_parameters2}
\begin{tabular}[t]{llllll}
\toprule
{\footnotesize{\textbf{Symbol}}} & {\footnotesize\textbf{Value}} & {\footnotesize\textbf{Reference}}\\ 
\midrule
{\footnotesize $t_0 $}
			  & {\footnotesize 02 March 2020} & {\footnotesize \cite{msas}}\\
			{\footnotesize $S_0$} &  {\footnotesize$7  857 353$} & {\footnotesize \cite{ansd}	}\\
              {\footnotesize $ 1/\beta_2$} &  {\footnotesize $1$ day} & {\footnotesize \cite{seydi3}} \\
              {\footnotesize $ 1/\beta_3$} &  {\footnotesize $10$ days} & {\footnotesize WHO} \\
			{\footnotesize $\alpha_1$} &{\footnotesize $30$ \%}  & {\footnotesize Assumed} \\
			{\footnotesize $1/ \eta$}&{\footnotesize $30$ days} & {\footnotesize Assumed} \\
			{\footnotesize $1/(\gamma+\eta)$}& {\footnotesize $15$} days  &   {\footnotesize Assumed}  \\
			{\footnotesize $K_{max}$} & {\footnotesize $3000$} & {\footnotesize Assumed} \\
			{\footnotesize $K_{0}$} & {\footnotesize $30$} & {\footnotesize Assumed} \\
			{\footnotesize $r$} & {\footnotesize $0.1$} & {\footnotesize Assumed} \\
			{\footnotesize $\epsilon$} & {\footnotesize $1000/S_0$} & {\footnotesize Assumed} \\
\bottomrule
\end{tabular}
\end{table}

\section{Parameter estimates using the early phase} \label{Sec3}
In this section, will use the approach developed in \cite{seydi1,seydi2,seydi3,seydi4} in order to estimate the parameters from the early phase of the epidemic. 
Let $C(t)$ be the cumulative number of reported cases at time $t$ be defined by
\begin{equation}\label{eq2.5}
C(t)=C(t_0)+(1-\alpha_1) \beta_3 \int_{t_0}^t  W(s) ds
\end{equation}
with $t_0=02$ Mars $2020$ the date of the first reported cases in Senegal so that $C(t_0)=1$.  The early phase of the epidemic corresponds to the period from the announcement of the first case on $02$ Mars $2020$ to the first government actions at $t_1=24$ March $2020$. In this phase we assume that the number of infected individuals growth exponentially. Hence with such assumption, the cumulative number of cases has the following form 
\begin{equation}\label{eq2.6}
C(t)=C(t_0)+\chi_{2} e^{\chi_{1}(t-t_0)}-\chi_2,\quad t\in [t_0,t_1].
\end{equation}
The Figure \ref{compar_data_cases} below shows the comparison between the $t \rightarrow C(t)$ and the data from $02$ March until $24$ March $2020$ that allows us to estimate the parameters $\chi_1$ and $\chi_2$ to 
\begin{equation}\label{eq2.7}
\chi_1=0.1612 \ \text{ and } \ \chi_2=1.979.
\end{equation} 
\begin{figure}[h!]
\begin{center}
\includegraphics[scale=0.45]{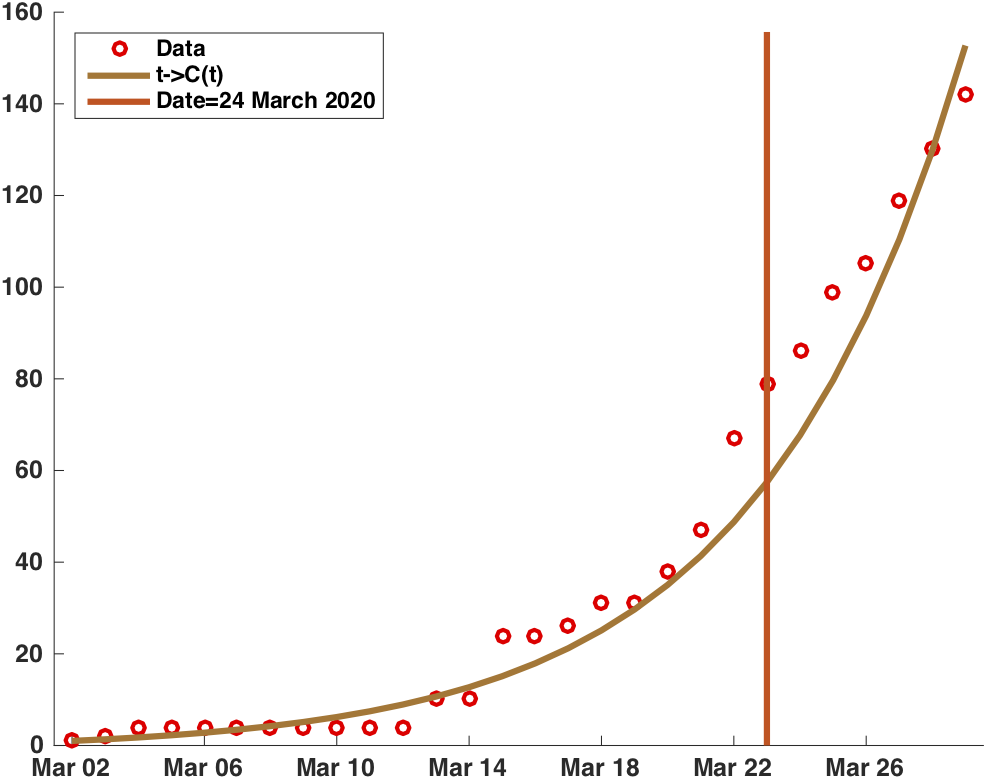}
\end{center}
\par\vspace{-0.25cm}
\caption{Comparison between the cumulative number of reported cases (data) and the map $t\rightarrow C(t)$ defined in (\ref{eq2.6}).}\label{compar_data_cases}
\end{figure}
\noindent Assuming that $S(t)$ is approximately equal to $S(t_0)$ in the short time interval  $[t_0,t_1]$ and that for each $t\in [t_0,t_1]$
$$ 
 E(t)=E_0 e^{\chi_1 (t-t_0)},\ W(t)=W_0 e^{\chi_1(t-t_0)}, \  I(t)=I_0 e^{\chi_1 (t-t_0)}
$$ 
and 
$$
\beta_1(t)=\beta_{10},
$$
it follows that we have the following approximation for a short time period 
\begin{equation} \label{eq2.8}
\left\{\begin{array}{llll}
\chi_1 E_0 &=&\b_{10}S_0 (W_0+\epsilon I_0)-\b_2 E_0\\
\chi_1 W_0 &=&\b_2 E_0 -\b_3 W_0\\
\chi_1 I_0 &=&(1-\a_1)\b_3 W_0-(\g +\eta) I_0 .
\end{array}
\right.
\end{equation}
Furthermore differentiating \eqref{eq2.5} and \eqref{eq2.6} with respect to $t$ we obtain 
$$
C'(t)=\chi_1 \chi_2 e^{\chi_1 (t-t_0)}=(1-\a_1)\b_3 W(t)=(1-\a_1)\b_3 W_0 e^{\chi_1(t-t_0)}
$$
that is 
\begin{equation}\label{eq2.9}
W_0=\dfrac{\chi_1 \chi_2}{(1-\a_1)\b_3}.
\end{equation}
Solving the second and the third equation of \eqref{eq2.8} we get:
\begin{equation}\label{eq2.10}
E_0=\dfrac{\chi_1+\beta_3}{\beta_2}W_0  \ \text{ and } \ I_0=\dfrac{(1-\a_1)\b_3}{\chi_1+\gamma+\eta} W_0.
\end{equation}
Plugging \eqref{eq2.9} and \eqref{eq2.10} into the first equation of \eqref{eq2.8}, it follows that
\begin{equation}\label{eq2.11}
\beta_{10}=\dfrac{\chi_1+\beta_2}{S_0(W_0+\epsilon I_0)}E_0.
\end{equation}
Therefore using \eqref{eq2.7} together with \eqref{eq2.9}-\eqref{eq2.11} and the values of the parameters of Table \ref{tab_parameters} we obtain 
\begin{equation}
E_0=1.18,\ W_0=4.55, \ I_0= 1.39
\end{equation}
and 
\begin{equation}
\beta_{10}=3.8593 \times   10^{-8}. 
\end{equation}
In order to obtain the values of the parameters $\theta_i$, $i=1,2,3,4$ that appear in the transmission rate $t\rightarrow \beta_1(t)$ we perform a curve fitting by using the data from $t_1=23$ March $2020$ until $28$ June $2020$. The estimated values are listed below
$$
\theta_1=5.1366, \    \theta_2=0.1224, \    \theta_3=1.7130, \  \text{and } \theta_4 =0.356.
$$
\section{The effective reproduction number} \label{Sec-Eff}
The effective reproduction number, $\mathcal{R}_{e}$, is the expected number of secondary cases produced by one typical infection joining a healthy population during its infectious period while the time-dependent effective reproduction number $\mathcal{R}_{e}(t)$ is the instantaneous transmissibility of the disease at time $t$. The time-dependent effective reproduction number allows to follow the evolution of the epidemic as the time evolves in particular how  $\mathcal{R}_{e}(t)$ is far above or  far below $1$.\\
In order to obtain $\mathcal{R}_e(t)$ we note that $\beta_1(t) S(t)$ is the number of new infections per unit of time generated by one infectious individual.  Thus assuming $\beta_1(t) S(t)$ approximately constant during $1/\beta_3$, it follows that the term 
$$
\frac{\beta_1(t) S(t)}{\beta_3}
$$ 
is the average number of new infections generated by one infectious individual (W) during its period of infectiousness $1/\beta_3$. Moreover one infectious individual (W) generates on average 
$$
(1-\a_1)\b_3 \times 1 \times \frac{1}{\beta_3}=1-\a_1
$$
infectious individual(s) (I). Those $1-\a_1$ infectious (I) has a maximum mean infectiousness period $1/(\gamma+\eta)$ so that they generate
$$
\epsilon \beta_1(t) S(t) \times (1-\a_1) \times \dfrac{1}{\gamma+\eta}
$$
new infections after $\dfrac{1}{\gamma+\eta}$ times.  \\
We finally define the time-dependent effective reproduction number as
\begin{equation}\label{Reff}
\mathcal{R}_e(t)=\dfrac{\beta_1(t) S(t)}{\beta_3}+\dfrac{(1-\a_1)\epsilon \beta_1(t) S(t)}{\gamma +\eta}. 
\end{equation}
\noindent The figure below gives the evolution of the effective reproduction number. 
\begin{figure}[h!]
\begin{center}
\includegraphics[scale=0.45]{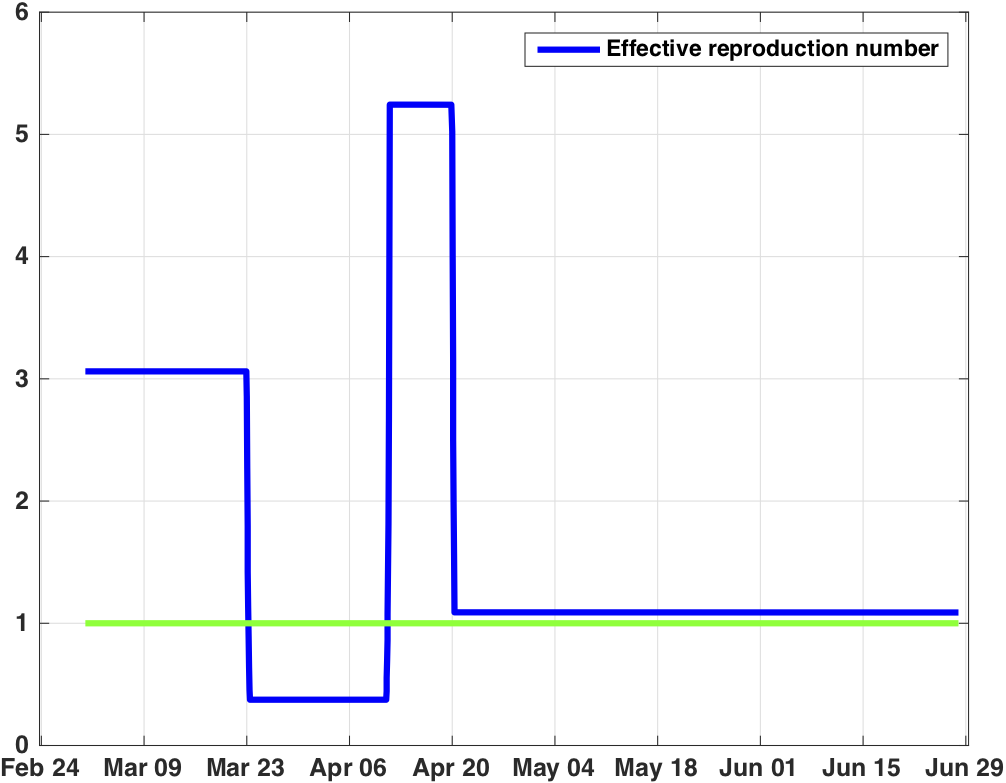}
\end{center}
\par\vspace{-0.25cm}
\caption{The effective reproduction $t\rightarrow \mathcal{R}_e(t)$ of the period from 02 March to 28 June 2020.}\label{Effective}
\end{figure}
\noindent We observe that the effective reproduction number is larger than $1$ (approximately $3$) during the early phase of the epidemic that is from 02 March to 24  March and is smaller than  $1$ during the second Phase. We can therefore observe that the first actions of the government had a positive effect on the spread of the epidemic because it tended to disappear if the effective  measures would have been kept  in the same trend. However, the sudden increase in cases on April 12th impacted the effective reproduction number to bring it to around 5. It goes down until it is close to but greater than 1 from April 20th. The cause of this increase is unknown to us but a better knowledge of the data may allow us to find an explanation.

\section{Numerical simulations}\label{Sec-Num}
In this section we perform numerical simulations of the dynamics of the spread of the COVID-19 epidemic based on our proposed model \eqref{eq:SEWIR-F}-\eqref{Ksa-sa}. We consider the regions of Dakar, Thies and Diourbel which concentrate around 95\% of the reported cases. To do so we assume that the initial number of susceptibles is the population of Dakar, Thies and Diourbel, that is to say $S_0=7  857 353$ \cite{ansd}. The values of the parameters of the model are listed in Table \ref{tab_parameters2}. The time-dependent transmission rate is given in \eqref{b2.3}-\eqref{b2.5} while the estimated initial conditions in Section \ref{Sec3} are  $E_0=1.18,\ W_0=4.55, \ I_0= 1.39$. In Figure  \ref{cumcases} below we compare the cumulative number of cases $t\rightarrow C(t)$ defined in \eqref{eq2.5} with $95 \%$ of the cumulative number of reported cases showing that our model agrees with the data. 
\begin{figure}[h!]
\begin{center}
\includegraphics[scale=0.45]{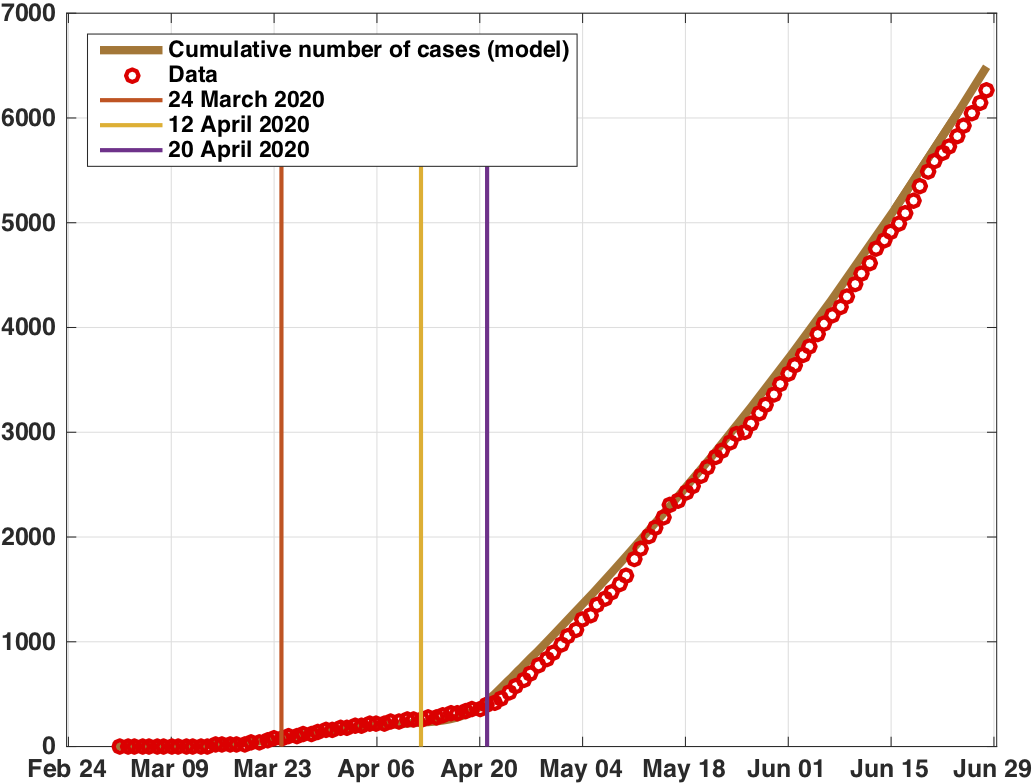}
\end{center}
\par\vspace{-0.25cm}
\caption{Comparison between 95 \% of the cumulative number of reported cases to the model for the population of Dakar, Thies and Diourbel.}\label{cumcases}
\end{figure}
\noindent In order to illustrate the impact of the saturation of health structures on the number of monitored individuals, we consider several scenarios. Indeed, we consider the same evolution of the reception capacity with interruptions in growth at predefined values. More precisely we consider the following dynamic
\begin{equation}\label{Ksa}
\dfrac{dK}{dt}=\left\lbrace
\begin{array}{llll}
r K(1-K/K_{max}) & \text{ if } & K \leq K_{sa}\\
0 & \text{ if } & K > K_{sa}\\
\end{array}
\right.
\end{equation}
where  $K_{sa}$, the saturation level, is respectively $3000$, $2500$, $2000$, $1500$ and $1000$ and is illustrated by the following figure. 
\begin{figure}[h!]
\begin{center}
\includegraphics[scale=0.45]{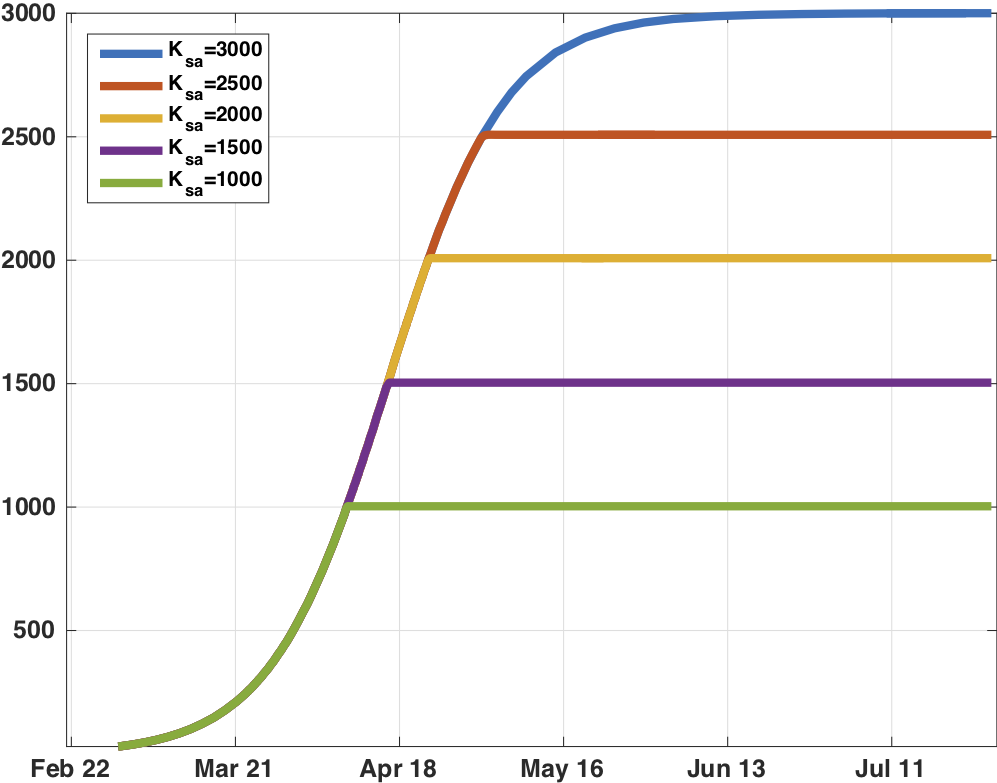}
\end{center}
\par\vspace{-0.25cm}
\caption{Evolution of the carrying capacity with saturation at $K_{sa}=3000$, $2500$, $2000$, $1500$ and $1000$.}\label{fig5:evol_Ksa}
\end{figure}
\noindent In Figure \ref{monitindiv1} we plot the corresponding number of monitored individuals that is $t\rightarrow I(t)$ when the dynamics of the carrying capacity is given by \eqref{Ksa} with $K_{sa}=3000$, $2500$, $2000$, $1500$ and $1000$ while the growth rate is fixed at $r=0.1$. We see that the epidemic is under control, until 28 June only when $K_{sa}=3000$ and $K_{sa}=2500$. This shows that the maximum value of the carrying capacity, $K_{sa}$, may impact significantly the efficiency of public health structures by slowing down the rate at which individuals are removed from the hospital. In Figure \ref{monitindiv2} we look at the influence of the growth rate $r$ on the number of monitored individuals. We see that increasing $r$ from $0.1$ to $0.5$, does not impact significantly the number of monitored individuals and the dates at which overwhelming occurs. This is the case for all values of $K_{sa}$. \\  
For  $r=0.05$, Figure  \ref{monitindiv3}, we see that the number of monitored individuals becomes larger than the maximum number of available resources at the early stage of the epidemic, around 20 April. We also observe that $I(t)$ remains unchanged for different values of $K_{sa}$. Thus our model captures and validates the obvious fact that the faster you increase the number of beds, the better you can handle the hospitalized people. Figure \ref{monitindiv3} emphasizes more the sensibility of the growth rate $r$ which can also be seen/interpreted as the timing in anticipating overwhelming. Indeed we observe that for the maximal value of resources $K_{sa}=K_{\max}=3000$, $K(t)$ stay larger than $I(t)$ until 28 June for $r\geq 0.06$ while for $r\leq 0.05$ we see that the number monitored becomes larger than available resources from mid-April to the end of June.
\begin{figure}[h!]
	\centering
	\subfloat[\ ]{%
		\includegraphics[scale=0.35]{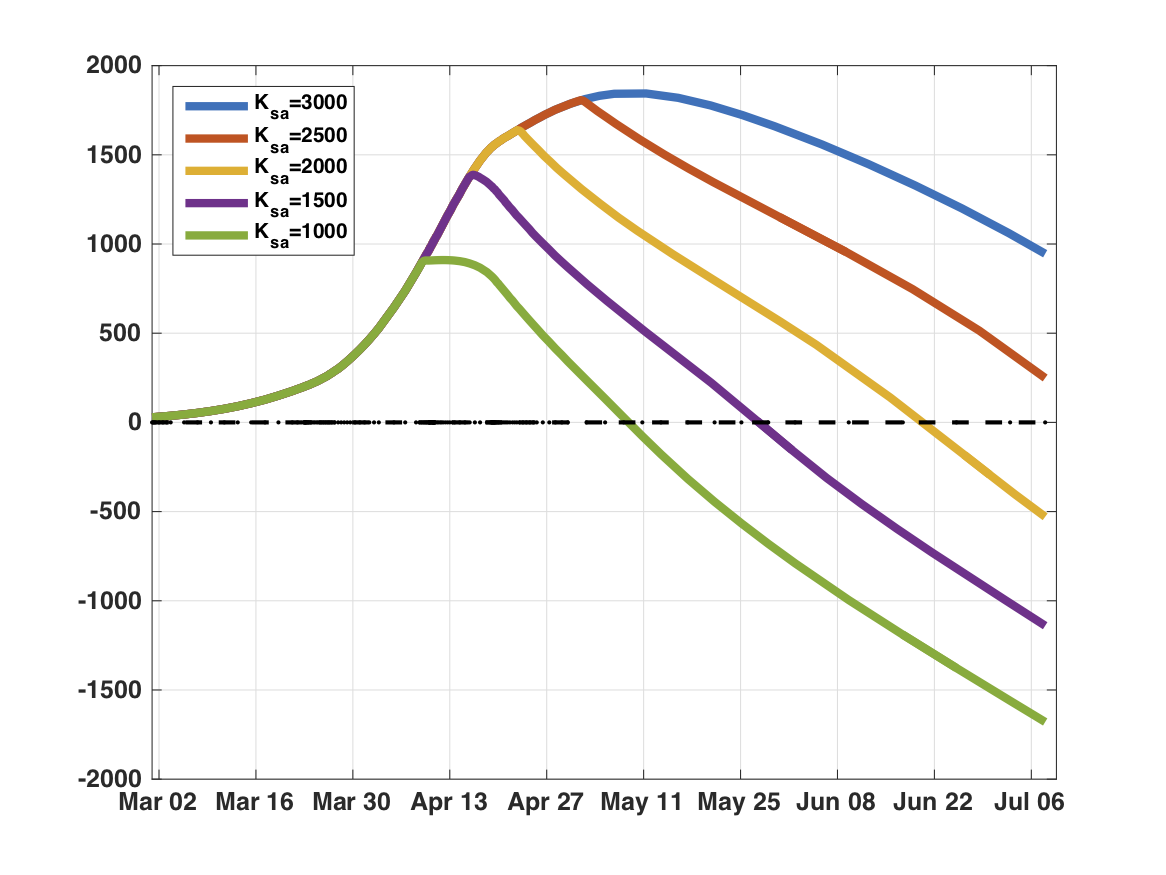}%
		\label{monitindiv1:a}%
	}%
	\subfloat[\ ]{%
		\includegraphics[scale=0.35]{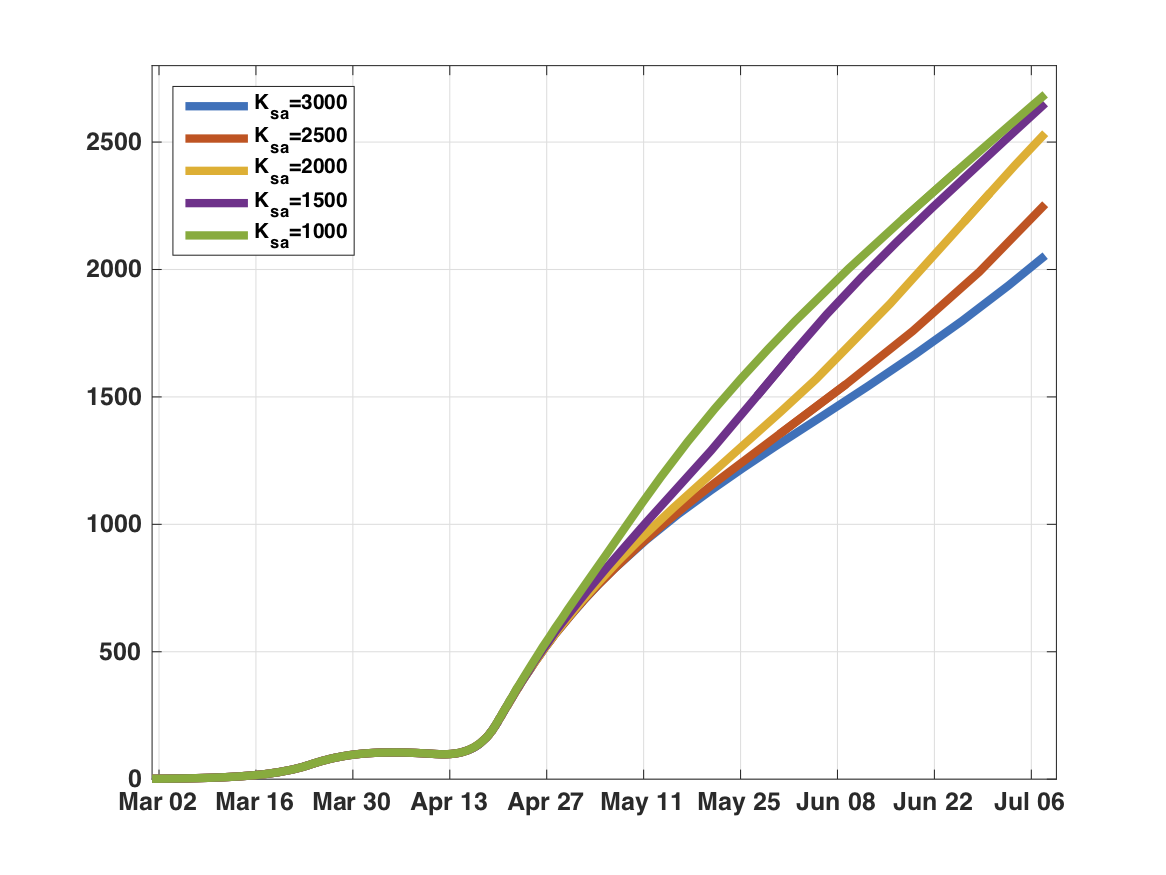}%
		\label{monitindiv1:b}%
	}
	\par\vspace{-0.25cm}
	\caption{\small{Number of monitored individuals with respect to the saturation carrying capacities values $K_{sa}=3000$, $2500$, $2000$, $1500$ and $1000$. The growth rate is $r=0.1$. Figure (A) shows the difference $K(t)-I(t)$ and Figure (B) shows the evolution of monitored individuals $I(t)$.
	}}
	\label{monitindiv1}%
\end{figure}
\begin{figure}[h!]
	\centering
	\subfloat[\ ]{%
		\includegraphics[scale=0.35]{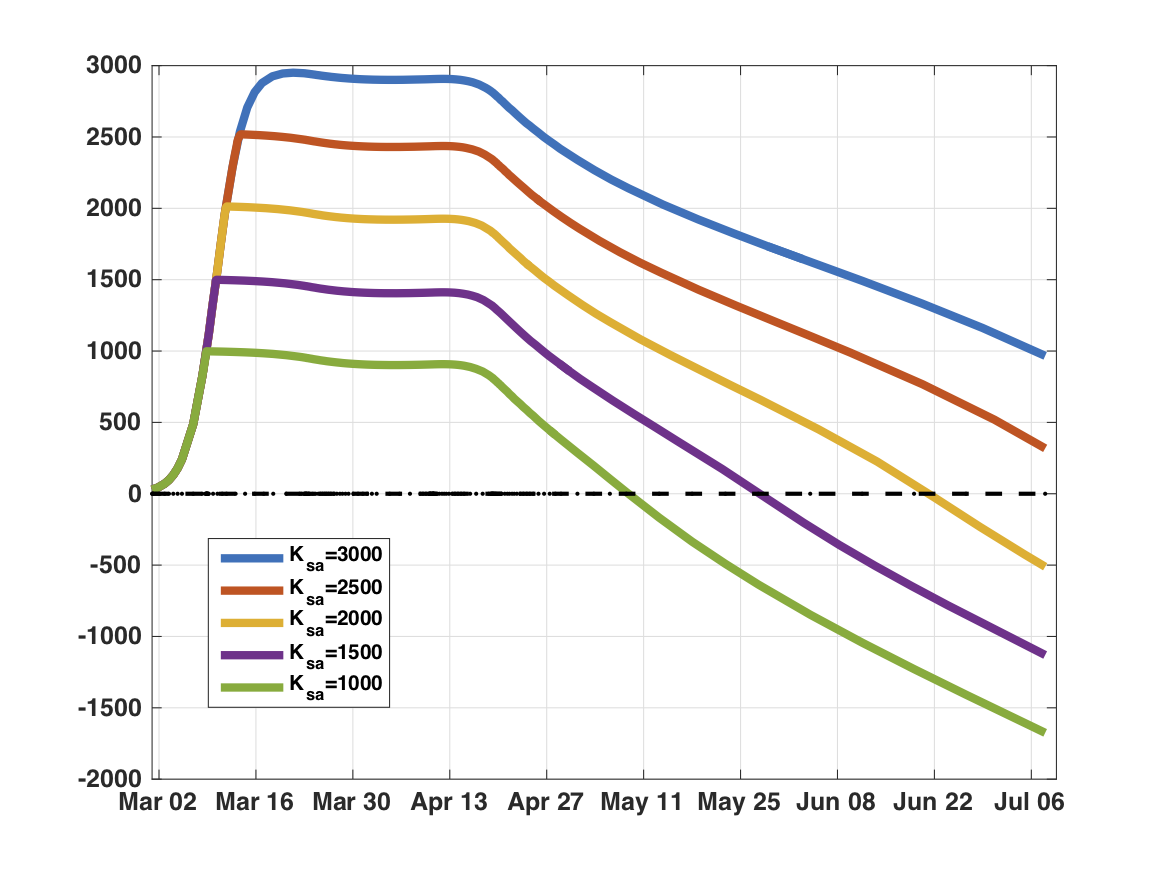}%
		\label{monitindiv2:a}%
	}%
	\subfloat[\ ]{%
		\includegraphics[scale=0.35]{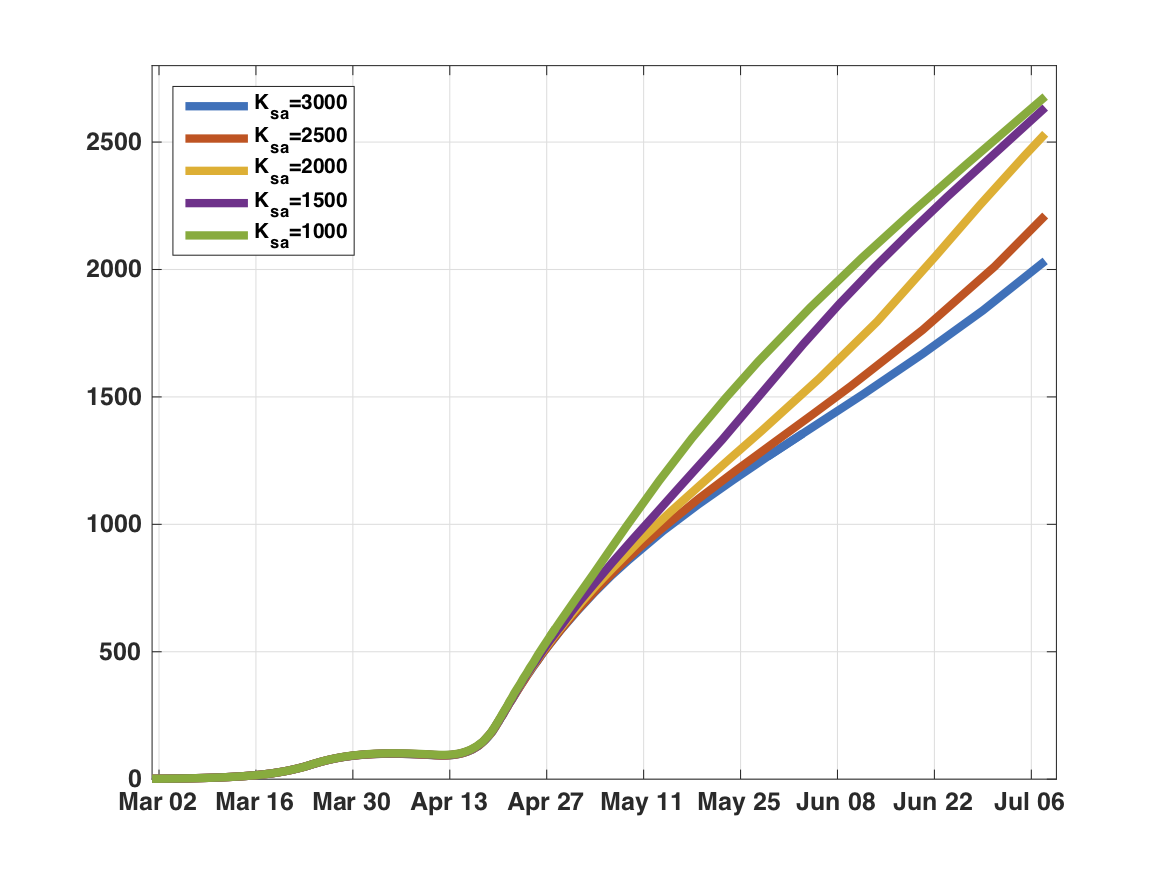}%
		\label{monitindiv2:b}%
	}
	\par\vspace{-0.25cm}
	\caption{\small{Number of monitored individuals with respect to the saturation carrying capacities values $K_{sa}=3000$, $2500$, $2000$, $1500$ and $1000$. The growth rate is $r=0.5$. Figure (A) shows the difference $K(t)-I(t)$ and Figure (B) shows the evolution of monitored individuals $I(t)$.
	}}
	\label{monitindiv2}
\end{figure}
\begin{figure}[h!]
	\centering
	\subfloat[\ ]{%
		\includegraphics[scale=0.35]{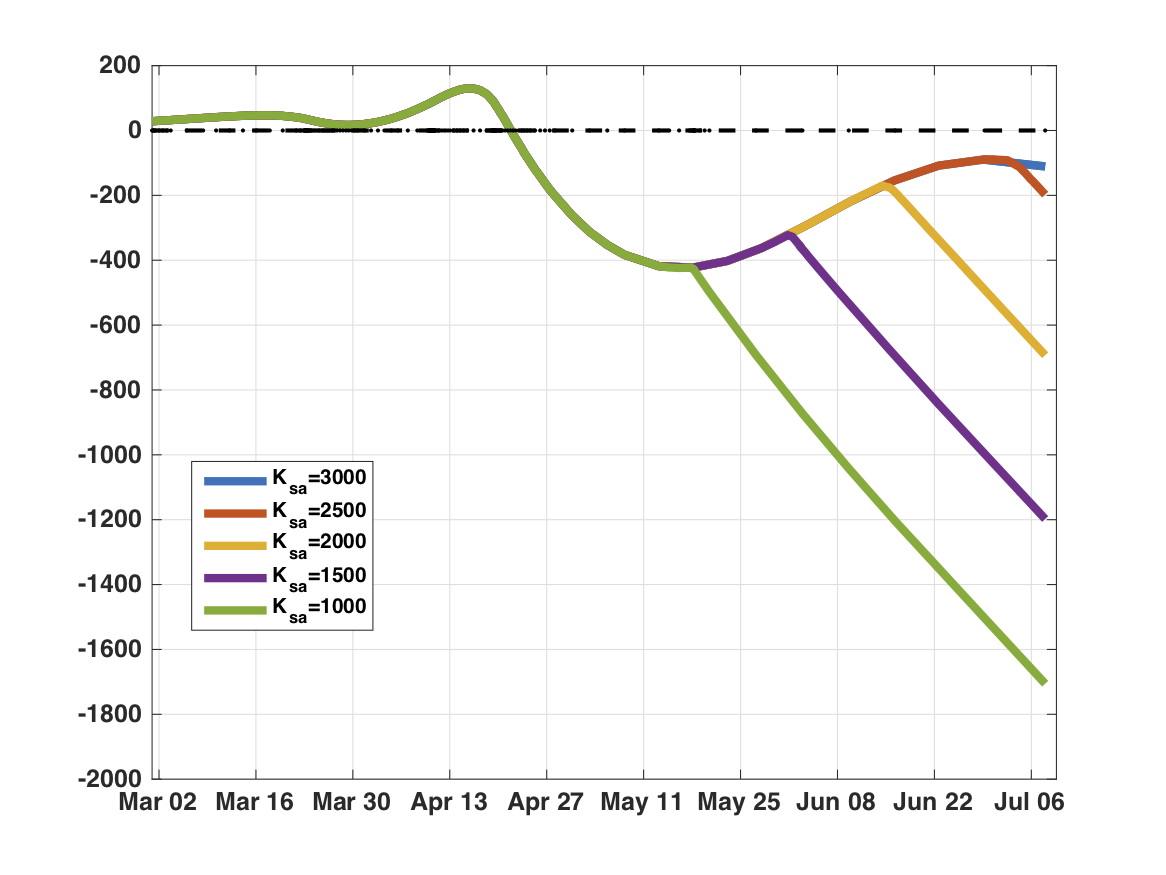}%
		\label{monitindiv3:a}%
	}%
	\subfloat[\ ]{%
		\includegraphics[scale=0.35]{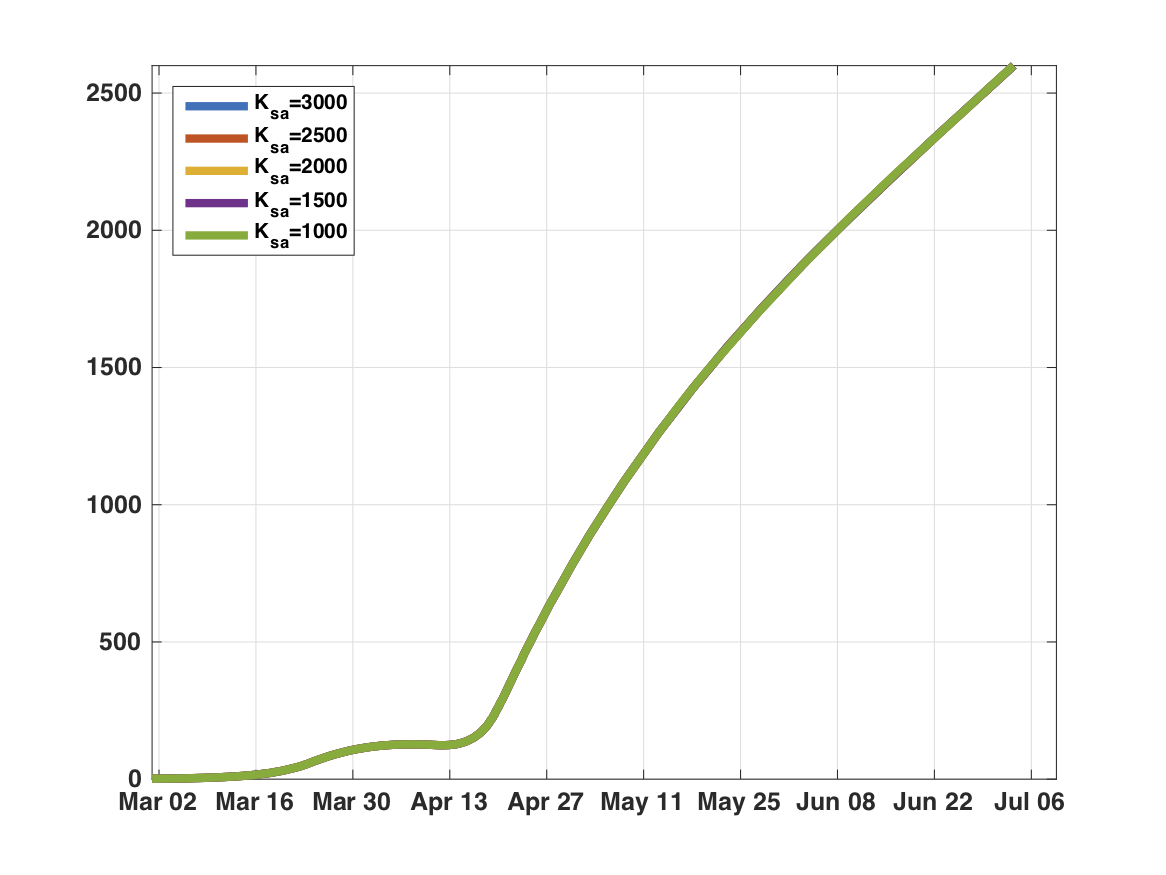}%
		\label{monitindiv3:b}%
	}
	\par\vspace{-0.25cm}
	\caption{\small{Number of monitored individuals with carrying capacities values $K_{sa}=3000$, $2500$, $2000$, $1500$ and $1000$. The growth rate is $r=0.05$. Figure (A) shows the difference $K(t)-I(t)$ and Figure (B) shows the evolution of the monitored individuals $I(t)$.%
	}}
	\label{monitindiv3}
\end{figure}
\begin{figure}[h!]
	\begin{center}
		\includegraphics[scale=0.45]{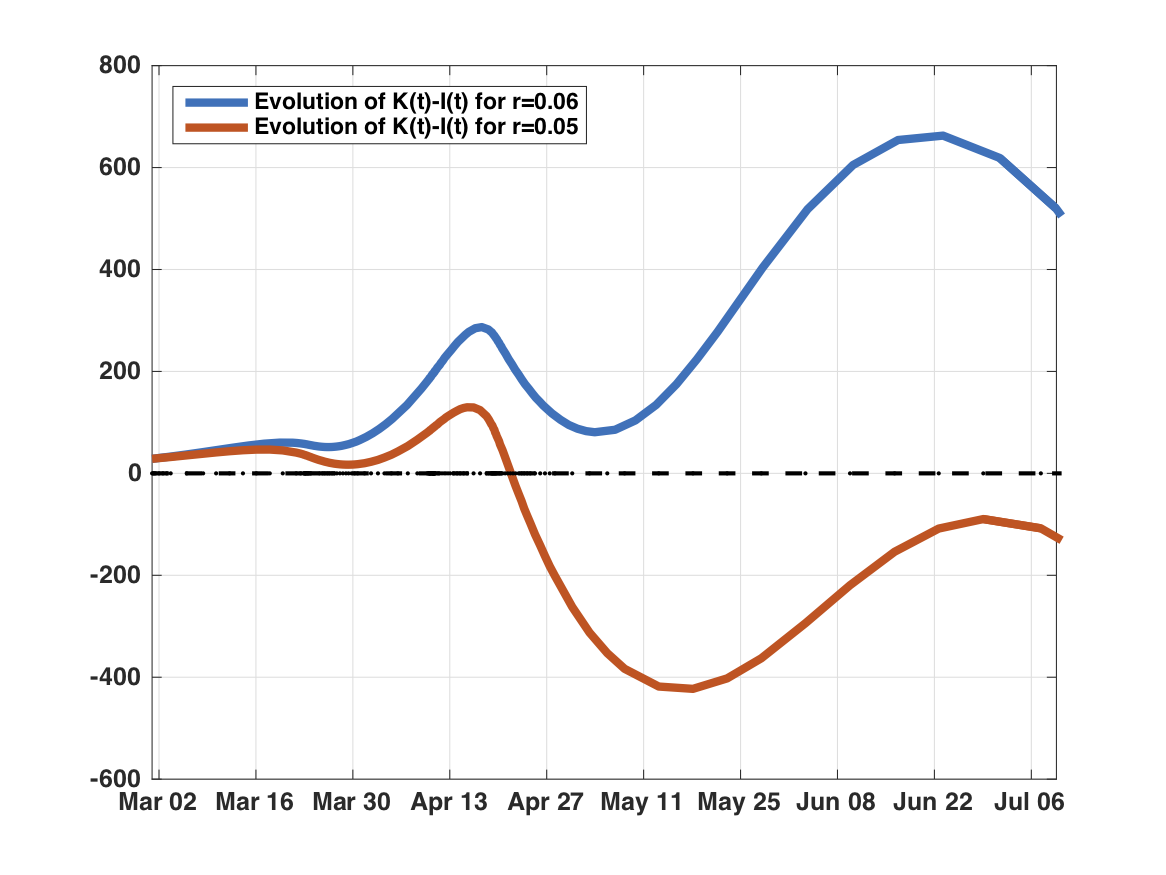}
	\end{center}
	\par\vspace{-0.35cm}
	\caption{Number of monitored individuals with carrying capacity values $K_{sa}=3000$. The growth rates are $r=0.05$ and $r=0.06$.}\label{monitindiv4}
\end{figure}
\ \\
\noindent In view of our discussion, it would be important to have a condition that guarantees the unsaturation of health structures. It is contained in the following
\begin{lem}
	Assume that $K_0>I_0$. Assume in addition that $t \to K(t)$ is increasing and differentiable with respect to $t$. If 
\begin{equation}\label{cc-New}
\dfrac{(1-\a_1)\b_3 W(t)}{\eta}<K(t),\ \forall t \in [t_0,T]
\end{equation}
	for some $T>0$ then we have 
	$$
	I(t)< K(t),\ \forall t\in [t_0,T].
	$$
\end{lem}
\begin{proof}
	Define 
	$$
	J(t):=K(t)-I(t),\ \forall t\geq t_0.
	$$
	Then, by assumption $J(t_0)=K_0-I_0>0$ and by continuity there is $t_1>t_0$ such that $J(t)>0$ for each $t \in [t_0,t_1]$. Define 
	$$
	\bar{t}:=\sup\left\lbrace t_0<t\leq T : J(s)>0,\ \forall s \in [t_0,t ] \right\rbrace. 
	$$
	Next, we show that $\bar{t}\geq T$. Assume by contradiction that $\bar{t}<T$. Then we have $J(\bar{t})=0$  that is  $I(\bar{t})=K(\bar{t})$ so that by (\ref{cc-New})
	$$
	J'(\bar{t})=K'(\bar{t})-I'(\bar{t})\geq \eta  I(\bar{t})-(1-\a_1)\b_3 W(\bar{t})=\eta  K(\bar{t})-(1-\a_1)\b_3 W(\bar{t})>0.
	$$
	This means that $J$ is locally strictly increasing from $\bar{t}$ which contradict the definition of $\bar{t}$. 
\end{proof}
\begin{rem}
Let us note that condition \eqref{cc-New} can be used in practice. Indeed $1/\eta$ is known and $(1-\a_1)\b_3 W(t)$ is the flux of new monitored individuals which is bounded above by the daily number of COVID-19 positive tests. Since $K(t)$ is a controllable quantity, the condition  \eqref{cc-New} may serve to anticipate an overwhelming over time.\\
It is important to note that in the above Lemma, we do not require $K(t)$ to have logistic growth. We only require $K(t)$ to be an increasing function. 
\end{rem}

\section{Machine learning for forecasting}\label{mlprophet}
In this section, we present a machine learning approach  for the  forecasting of the cumulative number of confirmed cases $C(t)$. 
First, we collect the pandemic data from \cite{msas}, from March 02, 2020, to October 22, 2020. Then, we perform a forecast with Prophet to predict the final size of coronavirus epidemy \\
The numerical tests are performed by using  Python with the Panda library 
\cite{python}, and were executed on a computer with the following characteristics: intel(R) Core-i7 CPU 2.60GHz, 24.0Gb of RAM, under the UNIX system.
\subsection{Prophet model}
Prophet \cite{prophet, sean},  is a procedure for forecasting time series data based on an additive model where non-linear trends are fit with yearly, weekly, and daily seasonality, plus holiday effects. It works best with time series that have strong seasonal effects and several seasons of historical data. Prophet is robust in dealing with missing data and shifts in the trend and typically handles well outliers. 
For the averaging method, the forecasts of all future values are equal to the average (or “mean”) of the historical data. If we let the historical data be denoted by $y_1,...,y_T$, then we can write the forecasts as
$$
\hat{y}_{T+h|T}=\bar{y}=(y_1+y_2+...+y_T)/T
$$
The notation $\hat{y}_{T+h|T}$ is a short-hand for the estimate of $y_{T+h}$  based on the data $y_1,...,y_T$.\\
A forecasting interval gives an interval within which we expect $y_t$  to lie on, with a specified probability. For example, assuming that the forecast errors are normally distributed, a 95\% forecasting interval for the  $h$-step forecast is 
$$
\hat{y}_{T+h|T}\pm1.96\hat{\sigma_h}
$$
where  ${\sigma_h}$ is an estimate of the standard deviation of the $h$-step  forecast distribution.  

\subsubsection{Diagnostics}\label{diagnostic_comp}
Here, we make some diagnostics by using the cross validation (see Table \ref{sn_crossvalid}) and the performance metrics (see Table \ref{sn_perfmetrics}) using mse, rmse, mae and mape. The Figure \ref{sn_cvm} illustrate these cross validation metrics, making 10 forecasts with cutoffs between 2020-10-08, 00:00:00 and 2020-10-17, 00:00:00  (initial='220 days', period='1 days', horizon = '5 days'). 
\par\vspace{-0.5cm}
\begin{figure}[h!]
  \subfloat[mape]{
	\begin{minipage}[1\width]{0.47\textwidth}
	   \centering
	   \includegraphics[width=1\textwidth]{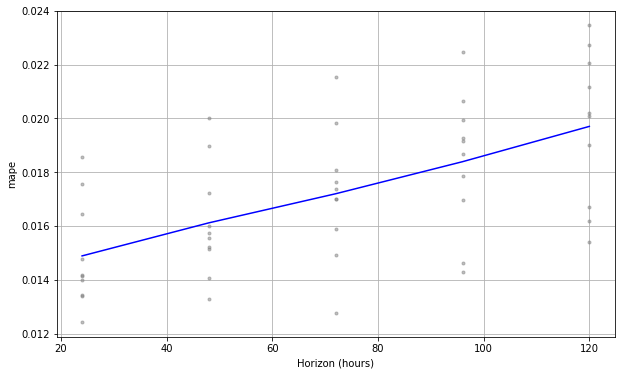}\label{sn_mape}
	\end{minipage}}
  \subfloat[mae]{
	\begin{minipage}[1\width]{ 0.47\textwidth}
	   \centering
	   \includegraphics[width=1.\textwidth]{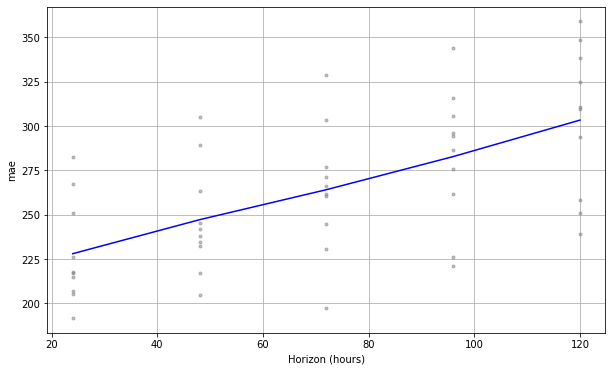}\label{sn_mae}
	\end{minipage}}
\newline
  \subfloat[rmse]{
	\begin{minipage}[1\width]{ 0.47\textwidth}
	   \centering
	   \includegraphics[width=1\textwidth]{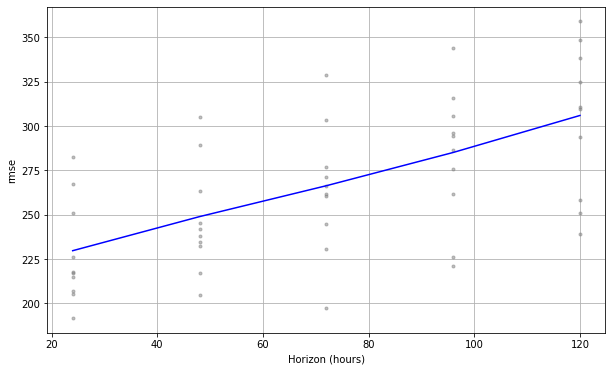}\label{sn_rmse}
	\end{minipage}}
%
  \subfloat[mse]{
	\begin{minipage}[1\width]{ 0.47\textwidth}
	   \centering
	   \includegraphics[width=1\textwidth]{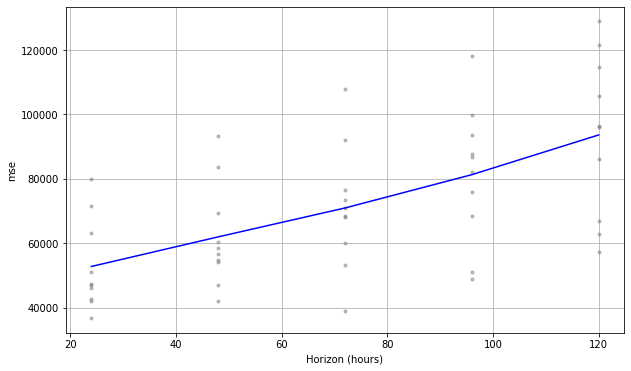}\label{sn_mse}
	\end{minipage}}
	\par\vspace{-0.25cm}
	\caption{Senegal: cross validation metrics}\label{sn_cvm}
\end{figure}
\begin{table}[h!] 
\begin{center}
\begin{tabular}{|c|c|c|c|c|c|} 
 \hline
          {\bf ds}   &         ${\bf \hat{y}}$ &    ${\bf\hat{y}_{lower}}$ &    ${\bf\hat{y}_{upper}}$  & ${\bf y} $	& {\bf cutoff} \\  \hline
 	2020-10-09    &	15495.491181 &	15364.133271& 	15630.001589& 	15213 &	2020-10-08\\ \hline
 	2020-10-10 &	15549.266664& 	15412.256057& 	15684.595605& 	15244& 	2020-10-08\\ \hline
 	2020-10-11 	&15596.622824 &	15460.701411& 	15730.823401 &	15268& 	2020-10-08\\ \hline
 	2020-10-12 &	15635.801860& 	15501.234167& 	15775.220030& 	15292& 	2020-10-08\\ \hline
 	2020-10-13 	&15666.025418& 	15524.679008& 	15804.770389 &	15307& 	2020-10-08\\ \hline         
\end{tabular}
\end{center}
\par\vspace{-0.25cm}
\caption{Senegal: cross validation}\label{sn_crossvalid}
\end{table}
\begin{table}[h!] 
\begin{center}
\begin{tabular}{|c|c|c|c|c|} 
 \hline
 	{\bf horizon} & {\bf mse } & {\bf rmse} 	&{\bf mae} &	{\bf mape}  \\  \hline
 	1 days &	52758.210719 &	229.691556 &	227.994396& 	0.014891 \\  \hline	
 	2 days &	61956.779147 &	248.911187 &		247.169729 &	0.016120 \\  \hline	
 	3 days& 	70940.217291 &	266.346048& 	264.104073& 	0.017201 \\  \hline	
 	4 days &	81284.637112 &	285.104607 &	282.785870 &	0.018392 \\  \hline	
 	5 days &	93627.316861 &	305.985812 &	303.367943 &	0.019704 \\  \hline	
\end{tabular}
\end{center}
\par\vspace{-0.25cm}
\caption{Senegal: performance metrics}\label{sn_perfmetrics}
\end{table}
\par
\noindent From Table \ref{sn_crossvalid}, by comparing the values obtained on October 13, 2020 (column $y$=15307) with the predicted one (column $ \hat{y}$=15666.025418), we see that the error is 2.29\%. The predicted value is always within the confidence interval. So, Prophet seems to give us good value.

\subsubsection{Trend changepoints and forecasting}\label{changepoints_sec} 
The rmse for Prophet Model is  53.358002. The Prophet forecasting of confirmed cases, with trend changepoints, is given by Figure \ref{sn_changepoints}, and the trends and weekly increase are given by Figure \ref{sn_procomponents}. \\
With Prophet, at $\sim$ November 12, 2020 we may obtain $>$ 16330 confirmed cases and  $>$ 16950 confirmed cases at $\sim$ December 01, 2020 (see Tables \ref{sn_1w_confcases} and \ref{sn_40d_confcases}, respectively). The forecastes of confirmed cases are illustrated in Figures \ref{prophet_3weeks_fore} and \ref{prophet_40days_fore}.
\begin{figure}[h!]
	\centering
	\includegraphics[scale=0.45]{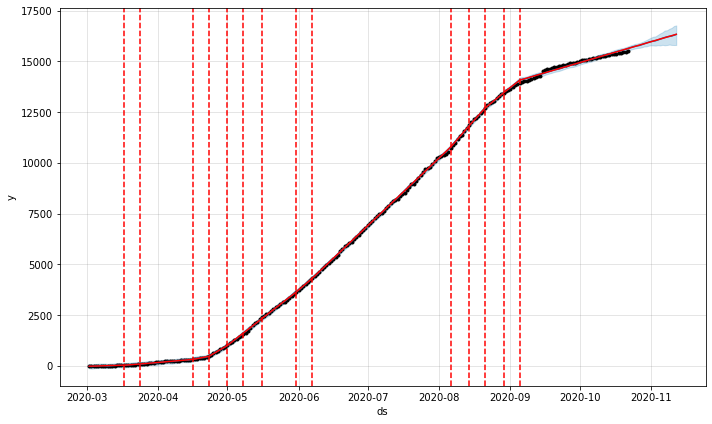}
	\par\vspace{-0.25cm}
	\caption{Senegal: changepoints of confirmed cases}\label{sn_changepoints}
\end{figure}
\begin{table}[h!]
\begin{center}
\begin{tabular}{|c|c|c|c|} 
 \hline
     {\bf ds}   &         ${\bf \hat{y}}$ &    ${\bf\hat{y}_{lower}}$ &    ${\bf\hat{y}_{upper}}$ \\  
           \hline
2020-11-08 &16213.793455& 15776.476419& 16625.216492\\ \hline
 2020-11-09& 16237.542798& 15730.499926& 16673.834380\\ \hline
 2020-11-10& 16254.670226& 15734.183777& 16750.503940\\ \hline
 2020-11-11& 16297.865380& 15743.971264& 16792.203760\\ \hline
 2020-11-12 &16335.788237 & 15755.096121& 16875.197916\\ \hline
 \end{tabular}
\end{center}
\par\vspace{-0.25cm}
\caption{Prophet: predicted cumulative confirmed cases $\sim$November 12, 2020.}\label{sn_1w_confcases}
\end{table} 
\begin{figure}[h!]
	\centering
	\includegraphics[scale=0.45]{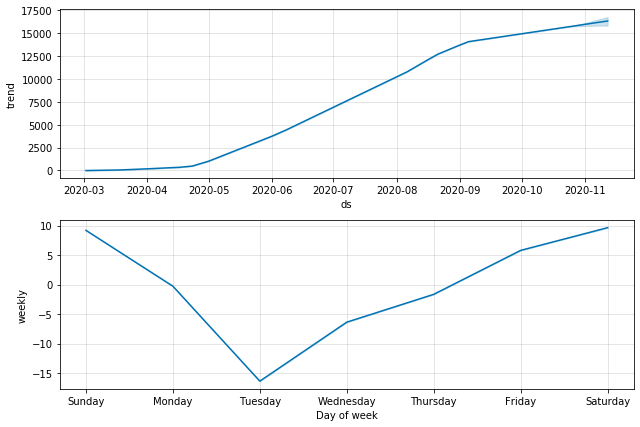}
	\par\vspace{-0.25cm}
	\caption{Senegal: Trends and weekly indrease of confirmed cases}\label{sn_procomponents}
\end{figure}
\begin{table}[h!]
\begin{center}
\begin{tabular}{|c|c|c|c|} 
 \hline
 {\bf ds}   &         ${\bf \hat{y}}$ &    ${\bf\hat{y}_{lower}}$ &    ${\bf\hat{y}_{upper}}$ \\  
           \hline
   2020-11-27 &16841.363557& 15724.240087 & 17997.233129\\ \hline
 2020-11-28& 16878.418793& 15728.625383& 18059.596350\\ \hline
 2020-11-29 &16911.177665 & 15665.479745 & 18162.293121\\ \hline
 2020-11-30& 16934.927009& 15683.579005 & 18193.215306\\ \hline
 2020-12-01 &16952.054437& 15654.901112 & 18268.706253\\ \hline         
\end{tabular}
\end{center}
\par\vspace{-0.25cm}
\caption{Prophet: predicted cumulative confirmed cases $\sim$December 01, 2020.}\label{sn_40d_confcases}
\end{table} 
\begin{figure}[h!]
  \subfloat[3 weeks forecasting]{
	\begin{minipage}[1.\width]{0.47\textwidth}
	   \centering
	   \includegraphics[width=1.1\textwidth]{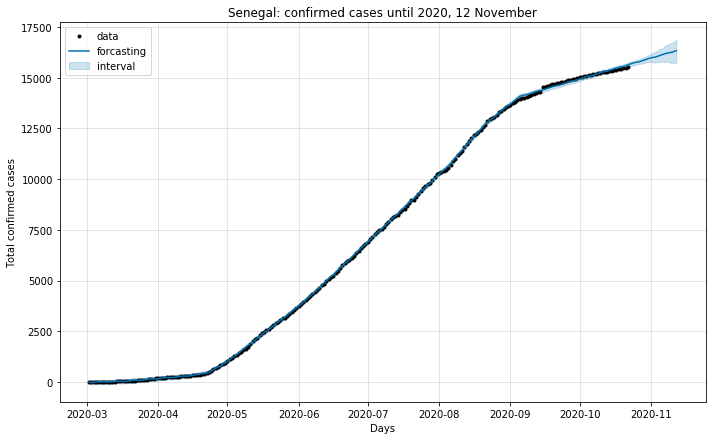}\label{prophet_3weeks_fore}
	\end{minipage}}
  \subfloat[40 days forecasting]{
	\begin{minipage}[1.\width]{ 0.47\textwidth}
	   \centering
	   \includegraphics[width=1.1\textwidth]{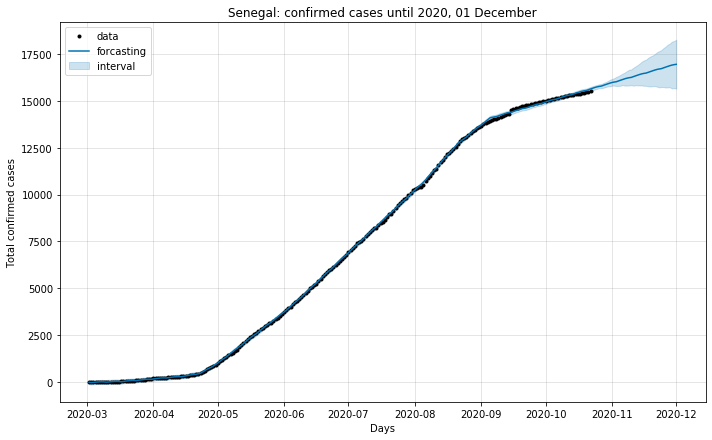}\label{prophet_40days_fore}
	\end{minipage}}
\par\vspace{-0.25cm}
	\caption{Senegal: Prophet for forecasting of confirmed cases}\label{sn_cvm}
\end{figure}

\end{document}